\documentclass[11pt]{llncs}
\usepackage{amsmath}
\usepackage{amssymb}
\usepackage{epic,gastex}
\usepackage{array}

\newcommand{\sa}{synchronizing automata}
\newcommand{\san}{synchronizing automaton}
\newcommand{\sw}{reset word}
\newcommand{\sws}{reset words}
\newcommand{\ssw}{reset word of minimum length}
\newcommand{\rl}{reset length}

\DeclareSymbolFont{rsfscript}{OMS}{rsfs}{m}{n}
\DeclareSymbolFontAlphabet{\mathrsfs}{rsfscript}

\begin{document}
\title{Slowly synchronizing automata and digraphs\thanks{Supported
by the Russian Foundation for Basic Research, grants 09-01-12142
and 10-01-00524, and by the Federal Education Agency of Russia,
grant 2.1.1/3537.}}

\titlerunning{Slowly synchronizing automata and digraphs}

\author{D. S. Ananichev \and V. V. Gusev \and M. V. Volkov}

\authorrunning{D. S. Ananichev, V. V. Gusev, M. V. Volkov}

\tocauthor{D. S. Ananichev, V. V. Gusev, M. V. Volkov
(Ekaterinburg, Russia)}

\institute{Department of Mathematics and Mechanics,\\
Ural State University, 620083 Ekaterinburg, RUSSIA\\
\email{Dmitry.Ananichev@usu.ru, vl.gusev@gmail.com,
Mikhail.Volkov@usu.ru}}

\maketitle

\begin{abstract}
We present several infinite series of \sa\ for which
the minimum length of reset words is close to the
square of the number of states. These automata are
closely related to primitive digraphs with large exponent.
\end{abstract}

\section{Background and overview}
\label{intro}

A \emph{complete deterministic finite automaton} (DFA) is a triple
$\mathrsfs{A}=\langle Q,\Sigma,\delta\rangle$, where $Q$ and
$\Sigma$ are finite sets called the \emph{state set} and the
\emph{input alphabet} respectively, and $\delta:Q\times\Sigma\to
Q$ is a totally defined function called the \emph{transition
function}. Let $\Sigma^*$ stand for the collection of all finite
words over the alphabet $\Sigma$, including the empty word.
The function $\delta$ extends to a function $Q\times\Sigma^*\to Q$
(still denoted by $\delta$) in the following natural way: for every
$q\in Q$ and $w\in\Sigma^*$, we set $\delta(q,w)=q$ if $w$ is empty
and $\delta(q,w)=\delta(\delta(q,v),a)$ if $w=va$ for some $v\in\Sigma^*$
and $a\in\Sigma$. Thus, via $\delta$, every word $w\in\Sigma^*$ acts
on the set $Q$.

A DFA $\mathrsfs{A}=\langle Q,\Sigma,\delta\rangle$ is called
\emph{synchronizing} if the action of some word $w\in\Sigma^*$
resets $\mathrsfs{A}$, that is, leaves the automaton in one
particular state no matter at which state in $Q$ it is applied:
$\delta(q,w)=\delta(q',w)$ for all $q,q'\in Q$. Any such word $w$
is said to be a \emph{reset word} for the DFA. The minimum length
of reset words for $\mathrsfs{A}$ is called the \emph{\rl} of
$\mathrsfs{A}$.

Synchronizing automata serve as transparent and natural models of
error-resistant systems in many applications (coding theory, robotics,
testing of reactive systems) and also reveal interesting connections
with symbolic dynamics and other parts of mathematics. For a brief
introduction to the theory of \sa\ we refer the reader to the recent
surveys~\cite{Sa05,Vo08}. Here we focus on the so-called \v{C}ern\'{y}
conjecture that constitutes a major open problem in this area.

In~1964 \v{C}ern\'{y}~\cite{Ce64} constructed for each $n>1$ a \san\
$\mathrsfs{C}_n$ with $n$ states whose reset length is $(n-1)^2$.
Soon after that he conjectured that these automata represent
the worst possible case, that is, every \san\ with $n$ states
can be reset by a word of length $(n-1)^2$. This simply looking conjecture
resists researchers' efforts for more than 40 years. Even though the
conjecture has been confirmed for various restricted classes of \sa\
(cf., e.g., \cite{Ep90,Du98,Ka03,Tr07,Tr08,AS09,Vo09}), no upper bound of magnitude
$O(n^2)$ for the \rl\ of $n$-state \sa\ is known in general. The best upper
bound achieved so far is $\frac{n^3-n}6$, see~\cite{Pi83}.

One of the difficulties that one encounters when approaching the \v{C}ern\'{y}
conjecture is that there are only very few \emph{extreme} automata, that is,
$n$-state \sa\ with reset length $(n-1)^2$. In fact, the \v{C}ern\'{y} series
$\mathrsfs{C}_n$ is the only known infinite series of extreme automata. Besides
that, only a few isolated examples of such automata have been found, see~\cite{Vo08}
for a complete list. Moreover, even \emph{slowly} \sa, that is, automata with reset
length close to the \v{C}ern\'{y} bound are very rare. This empirical observation
is supported also by probabilistic arguments. For instance, the probability that
a composition of $2n$ random self-maps of a set of size $n$ is a constant map
tends to~1 as $n$ goes to infinity~\cite{Hi88}. In terms of automata, this
result means that the \rl\ of a random automaton with $n$ states and at least
$2n$ input letters does not exceed $2n$. For further results of the same
flavor see~\cite{SZ}. Thus, there is no hope to find new examples of slowly
\sa\ by a lucky chance or via a random sampling experiment.

We therefore have designed and performed a set of exhaustive search
experiments. Our experiments are briefly described in Section~\ref{experiments}
while the main body of the paper is devoted to a theoretical analysis
of their outcome. We concentrate on two principal issues. In Section~\ref{matrices}
we discuss a similarity between the distribution of reset lengths of \sa\
and the distribution of exponents of primitive digraphs. Section~\ref{sss}
presents a few series of slowly \sa. Most of these series have been expanded
from new examples discovered in the course of our experiments. In our opinion,
the proof technique is also of interest; in fact, we provide a transparent
and uniform approach to all sufficiently large slowly \sa\ with 2~input letters,
both new and already known ones.

\section{Preliminaries}
\label{preliminaries}

We start with recalling two elementary and well-known number-theoretic results.
\begin{lemma}[{\mdseries\cite[Theorem 1.0.1]{RaAl05}}]
\label{schur}
If $k_1,\dots,k_m$ are positive integers whose greatest common divisor is equal
to\/ $1$, then there exists an integer $N$ such that every integer larger than $N$
is expressible as a non-negative integer combination of $k_1,\dots,k_m$.
\end{lemma}
The question of how the least $N$ with the property stated in Lemma~\ref{schur}
depends on the integers $k_1,\dots,k_m$ is known as the \emph{diophantine Frobenius
problem} and in general is highly non-trivial, see~\cite{RaAl05}. There is, however,
a simple special case which we will need in Section~\ref{sss}.
\begin{lemma}[{\mdseries\cite[Theorem 2.1.1]{RaAl05}}]
\label{sylvester}
If $k_1,k_2$ are relatively prime positive integers, then $k_1k_2-k_1-k_2$ is
the largest integer that is not expressible as a non-negative integer combination
of $k_1$ and $k_2$.
\end{lemma}

A \emph{directed graph} (digraph) is a pair $D=\langle V,E\rangle$ where $V$ is
a finite set and $E\subseteq V\times V$. We refer to elements of $V$ and $E$ as
\emph{vertices} and \emph{edges}. Observe that our definition allows loops but
excludes multiple edges. If $v,v'\in V$ and $e=(v,v')\in E$, the edge $e$ is said
to be \emph{outgoing} for $v$. We assume the reader's acquaintance with basic
notions of the theory of directed graphs such as (directed) path,
cycle, isomorphism etc.

Given a DFA $\mathrsfs{A}=\langle Q,\Sigma,\delta\rangle$, its \emph{underlying
digraph} $D(\mathrsfs{A})$ has $Q$ as the vertex set and $(q,q')\in Q\times Q$ is
an edge of $D(\mathrsfs{A})$ if and only if $q'=\delta(q,a)$ for some $a\in\Sigma$.
It is easy to see that a digraph $D$ is isomorphic to the underlying digraph of some
DFA if and only if each vertex of $D$ has at least one outgoing edge. In the sequel,
we always consider only digraphs satisfying this property. Every DFA $\mathrsfs{A}$
such that $D\cong D(\mathrsfs{A})$ is called a \emph{coloring} of $D$. Thus, every
coloring of $D$ is defined by assigning non-empty sets of labels (colors) from some
alphabet $\Sigma$ to edges of $D$ such that the label sets assigned to the outgoing
edges of each vertex form a partition of $\Sigma$. Fig.\,\ref{fig:cerny} shows a
digraph and two of its colorings by $\Sigma=\{a,b\}$.
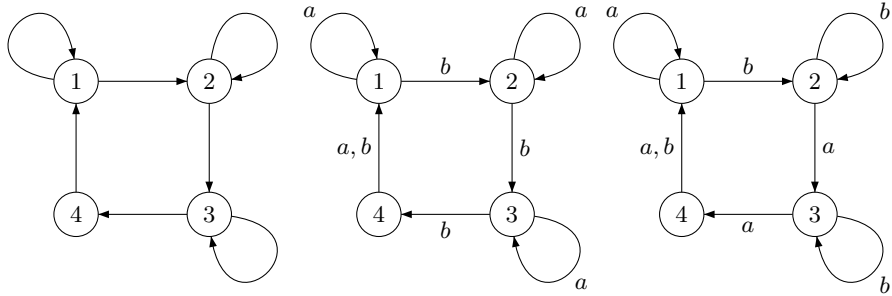
\begin{figure}[ht]
 \begin{center}
  \unitlength=2.8pt
    \begin{picture}(18,26)(-60,-4)
    \gasset{Nw=6,Nh=6,Nmr=3}
    \node(A1)(0,18){$1$}
    \node(A2)(18,18){$2$}
    \node(A3)(18,0){$3$}
    \node(A4)(0,0){$4$}
    \drawloop[loopangle=135](A1){$a$}
    \drawloop[loopangle=45](A2){$b$}
    \drawloop[loopangle=-45](A3){$b$}
    \drawedge(A1,A2){$b$}
    \drawedge(A2,A3){$a$}
    \drawedge(A3,A4){$a$}
    \drawedge(A4,A1){$a,b$}
    \end{picture}
 \begin{picture}(18,26)(0,-4)
    \gasset{Nw=6,Nh=6,Nmr=3}
    \node(A1)(0,18){$1$}
    \node(A2)(18,18){$2$}
    \node(A3)(18,0){$3$}
    \node(A4)(0,0){$4$}
    \drawloop[loopangle=135](A1){$a$}
    \drawloop[loopangle=45](A2){$a$}
    \drawloop[loopangle=-45](A3){$a$}
    \drawedge(A1,A2){$b$}
    \drawedge(A2,A3){$b$}
    \drawedge(A3,A4){$b$}
    \drawedge(A4,A1){$a,b$}
    \end{picture}
 \begin{picture}(18,26)(60,-4)
    \gasset{Nw=6,Nh=6,Nmr=3}
    \node(A1)(0,18){$1$}
    \node(A2)(18,18){$2$}
    \node(A3)(18,0){$3$}
    \node(A4)(0,0){$4$}
    \drawloop[loopangle=135](A1){}
    \drawloop[loopangle=45](A2){}
    \drawloop[loopangle=-45](A3){}
    \drawedge(A1,A2){}
    \drawedge(A2,A3){}
    \drawedge(A3,A4){}
    \drawedge(A4,A1){}
    \end{picture}
 \end{center}
 \caption{A digraph and two of its colorings}
 \label{fig:cerny}
\end{figure}

The \emph{matrix} of a digraph $D=\langle V,E\rangle$ is just the incidence
matrix of the edge relation, that is, a $V\times V$-matrix whose entry in
the row $v$ and the column $v'$ is 1 if $(v,v')\in E$ and 0 otherwise. For
instance, the matrix of the digraph in Fig.~1 (with respect to the chosen
numbering of its vertices) is $\left(\begin{smallmatrix}1&1&0&0\\
0&1&1&0\\ 0&0&1&1\\ 1&0&0&0 \end{smallmatrix}\right)$. Conversely, given
an $n\times n$-matrix $P=(p_{ij})$ with non-negative real entries, we assign
to it a digraph $D(P)$ on the set $\{1,2,\dots,n\}$ as follows: $(i,j)$ is an
edge of $D(P)$ if and only if $p_{ij}>0$. This ``two-way'' correspondence
allows us to formulate in terms of digraphs several important for the sequel
notions and results which originated in the classical Perron--Frobenius theory
of non-negative matrices.

Recall that a digraph $D=\langle V,E\rangle$ is said to be \emph{strongly
connected} if for every pair $(v,v')\in V\times V$, there exists a path
from $v$ to $v'$. By the $t^{th}$ \emph{power}
of $D$ we mean the digraph $D^t$ with the same vertex set $V$, such that
$(v,v')\in V\times V$ is an edge of $D^t$ if and only if there is a path
in $D$ from $v$ to $v'$ of length precisely $t$. If $M$ is the matrix of $D$,
then the digraph $D^t$ can be equivalently defined as $D(M^t)$, where
$M^t$ is the usual $t^{th}$ power of $M$.

A strongly connected digraph $D$ is called \emph{primitive} if the greatest
common divisor of the lengths of all cycles in $D$ is equal to~1. (In the
literature such graphs are sometimes called \emph{aperiodic}.)
Lemma~\ref{schur} readily implies that if $D$ is a primitive digraph, then
in some power $D^t$ of $D$ every pair of vertices constitutes an edge, i.e.,
$D^t$ is a complete digraph with loops. (This is equivalent to saying that
every entry of the matrix $M^t$, where $M$ is the matrix of $D$, is positive.)
The least $t$ with this property is called the \emph{exponent} of the digraph $D$
and is denoted by $\gamma(D)$. We need some results on exponents of digraphs
summarized as follows.

\begin{theorem}
\label{dulmage}
\emph{(a) (Wielandt's theorem, see \cite{Wi50,DM62}, \cite[Theorem~1]{DM64})}
If a primitive graph $D$ has $n$~vertices, then  $\gamma(D)\le(n-1)^2+1$.

\emph{(b) \cite[Theorem~6 and Corollary 4]{DM64}} Up to isomorphism, there is
exactly one primitive digraph $D$ on $n>2$ vertices with $\gamma(D)=(n-1)^2+1$,
and exactly one with $\gamma(D)=(n-1)^2$. The matrices of the digraphs are
\begin{equation}
\label{wielandt}
\begin{pmatrix}
0 & 1 & 0 & \dots & 0 & 0\\
0 & 0 & 1 & \dots & 0 & 0\\
\hdotsfor{6}\\
0 & 0 & 0 & \dots & 0 & 1\\
1 & 1 & 0 & \dots & 0 & 0
\end{pmatrix} \text{ and }
\begin{pmatrix}
0 & 1 & 0 & \dots & 0 & 0\\
0 & 0 & 1 & \dots & 0 & 0\\
\hdotsfor{6}\\
1 & 0 & 0 & \dots & 0 & 1\\
1 & 1 & 0 & \dots & 0 & 0
\end{pmatrix}
\text{ respectively.}
\end{equation}

\emph{(c) \cite[Theorem~7]{DM64}} If $n>4$ is even, then there is no
primitive digraph $D$ on $n$ vertices such that $n^2-4n+6<\gamma(D)<(n-1)^2$,
and, up to isomorphism, there are either $3$ or $4$ primitive
digraphs $D$ on $n$ vertices with $\gamma(D)=n^2-4n+6$, according
as $n$ is or is not a~multiple of $3$.

\emph{(d) \cite[Theorem~8]{DM64}} If $n>3$ is odd, then there is no
primitive digraph $D$ on $n$ vertices such that $n^2-3n+4<\gamma(D)<(n-1)^2$,
and, up to isomorphism, there is exactly one primitive digraph $D$ on $n$ vertices
with $\gamma(D)=n^2-3n+4$, exactly one with $\gamma(D)=n^2-3n+3$, and
exactly two with $\gamma(D)=n^2-3n+2$. The matrices of these digraphs are:
\begin{equation}
\label{odd island}
\begin{pmatrix}
0 & 1 & 0 & \dots & 0 & 0\\
0 & 0 & 1 & \dots & 0 & 0\\
\hdotsfor{6}\\
0 & 0 & 0 & \dots & 1 & 0\\
0 & 0 & 0 & \dots & 0 & 1\\
1 & 0 & 1 & \dots & 0 & 0
\end{pmatrix}\!,
\begin{pmatrix}
0 & 1 & 0 & \dots & 0 & 0\\
0 & 0 & 1 & \dots & 0 & 0\\
\hdotsfor{6}\\
0 & 0 & 0 & \dots & 1 & 0\\
0 & 1 & 0 & \dots & 0 & 1\\
1 & 0 & 1 & \dots & 0 & 0
\end{pmatrix}\!,
\begin{pmatrix}
0 & 1 & 0 & \dots & 0 & 0\\
0 & 0 & 1 & \dots & 0 & 0\\
\hdotsfor{6}\\
1 & 0 & 0 & \dots & 1 & 0\\
0 & 1 & 0 & \dots & 0 & 1\\
1 & 0 & 1 & \dots & 0 & 0
\end{pmatrix}\!,
\begin{pmatrix}
0 & 1 & 0 & \dots & 0 & 0\\
0 & 0 & 1 & \dots & 0 & 0\\
\hdotsfor{6}\\
1 & 0 & 0 & \dots & 1 & 0\\
0 & 0 & 0 & \dots & 0 & 1\\
1 & 0 & 1 & \dots & 0 & 0
\end{pmatrix}\!.
\end{equation}

\emph{(e) \cite[Theorem~8]{DM64}} If $n>3$ is odd, then there is no
primitive digraph $D$ on $n$ vertices such that $n^2-4n+6<\gamma(D)<n^2-3n+2$,
and, up to isomorphism, there are either $3$ or $4$ primitive digraphs $D$
on $n$ vertices with $\gamma(D)=n^2-4n+6$, according as $n$ is or is not
a~multiple of $3$.
\end{theorem}

\section{Exponents of digraphs vs lengths of reset words}
\label{matrices}

As mentioned in Section~\ref{intro}, this paper has grown
from certain observations made when we analyzed experimental
results. One such observation has been a similarity between
the ``upper parts'' of two sequences: the sequence of possible
reset lengths of 2-letter \sa\ with $n$ states and the sequence
of possible exponents of primitive digraphs with $n$ vertices.
As it is clear from Theorem~\ref{dulmage}, the upper part of the
latter sequence has certain gaps whose sizes and positions depend
on the parity of $n$; our experiments have revealed a similar pattern
of gaps in the upper part of the former sequence. Table~\ref{9 states}
illustrates this observation for $n=9$.

\begin{table}[hbt]
\extrarowheight=1pt
\caption{Exponents of primitive digraphs with $9$ vertices
vs lengths of shortest reset words for 2-letter \sa\ with $9$ states}\label{9 states}
\begin{tabular}{|p{5.4cm}||c|c|c|c|c|c|c|c|c|c|c|c|c|c|c|}
\hline
\centering{$N$} & 65 & 64 & 63 & 62 & 61 & 60 & 59 & 58 & 57 & 56 & 55 & 54 & 53 & 52 & 51 \\
\hline
\raggedright{Number of non-isomorphic primitive digraphs
with exponent} $N$
& \raisebox{-6pt}{1} & \raisebox{-6pt}{1} & \raisebox{-6pt}{0} & \raisebox{-6pt}{0} & \raisebox{-6pt}{0}
& \raisebox{-6pt}{0} & \raisebox{-6pt}{0} & \raisebox{-6pt}{1} & \raisebox{-6pt}{1} & \raisebox{-6pt}{2}
& \raisebox{-6pt}{0} & \raisebox{-6pt}{0} & \raisebox{-6pt}{0} & \raisebox{-6pt}{0} & \raisebox{-6pt}{4} \\
\hline
\raggedright{Number of non-isomorphic 2-letter \sa\ with
reset length $N$}
&\raisebox{-11pt}{0} &\raisebox{-11pt}{1} &\raisebox{-11pt}{0} &\raisebox{-11pt}{0} &\raisebox{-11pt}{0}
&\raisebox{-11pt}{0} &\raisebox{-11pt}{0} &\raisebox{-11pt}{1} &\raisebox{-11pt}{2} &\raisebox{-11pt}{3}
&\raisebox{-11pt}{0} &\raisebox{-11pt}{0} &\raisebox{-11pt}{0} &\raisebox{-11pt}{4} &\raisebox{-11pt}{4} \\
\hline
\end{tabular}
\end{table}
The data in the second row of Table~\ref{9 states} are calculated
from Theorem~\ref{dulmage}, while the data in the third row come from
our experiments.

Concerning gaps in the upper part of the sequence of possible reset lengths of
2-letter \sa\ with a given number of states, we notice that the first gap was
registered in earlier experiments. (Namely, according to~\cite{Tr06,Tr06a},
for $n=7,8,9,10$ there exists no 2-letter \sa\ with $n$ states with reset lengths
between $n^2-2n$ and $n^2-3n+5$.) However, to the best of our knowledge, the second
gap as seen in Table~\ref{9 states} has not been reported in the literature up to now.

We strongly believe that the observed similarity is more than a coincidence.
Clearly, there are deep connections between primitive digraphs and \sa. Indeed,
it is well known (see~\cite{AGW}) that if the underlying digraph of a \san\
is strongly connected that the digraph must be primitive; on the other hand,
as follows from Trahtman's proof~\cite{Tr09} of the so-called Road Coloring
conjecture by Adler, Goodwyn, and Weiss~\cite{AGW}, every primitive digraph
admits a synchronizing coloring. This, however, does not suffice to explain
similarities such as in Table~\ref{9 states} because many of slowly \sa\
``responsible'' for non-zero entries in the third row cannot be obtained
as colorings of primitive digraphs with large exponents corresponding
to non-zero entries in the second row. In the next section we demonstrate
some new connections between primitive digraphs with large exponents and
slowly \sa\ with two input letters. In this way, we derive all known
series of such automata and construct many new ones.

\section{Some series of slowly \sa}
\label{sss}

Due to space limitations, we present here only a part of our results on slowly \sa.
Namely, we restrict ourselves to series derived from three of the primitive digraphs
whose matrices are listed in Theorem~\ref{dulmage}. These series, in particular,
ensure that the ``island'' of reset lengths between $n^2-3n+2$ and
$n^2-3n+4$ exists for each $n$.

We start with the digraph $W_n$ corresponding to the first matrix in~\eqref{wielandt}.
The digraph (more precisely, its matrix) first appeared in Wielandt's seminal paper~\cite{Wi50}.
It has $n$ vertices $1,2,\dots,n$, say, and the following $n+1$ edges: $(i,i+1)$ for
$i=1,\dots,n-1$, $(n,1)$, and $(n,2)$.

It is easy to see that, up to isomorphism and renaming of letters, there exists
a unique coloring of the digraph $W_n$ by two letters. Let $\mathrsfs{W}_n$ denote
this coloring. Fig.\,\ref{fig:anan} shows the digraph $W_n$ and the DFA $\mathrsfs{W}_n$.
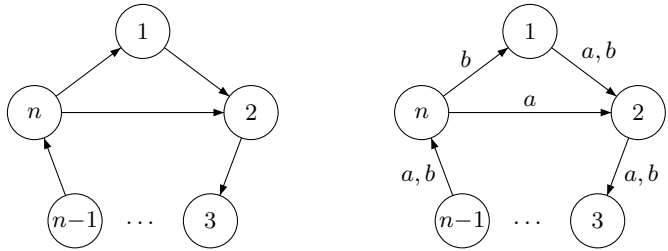
\begin{figure}[ht]
\begin{center}
\unitlength .45mm
\begin{picture}(72,56)(20,-72)
\gasset{Nw=16,Nh=16,Nmr=8}
\node(n0)(36.0,-16.0){1}
\node(n1)(4.0,-40.0){$n$} \node(n2)(68.0,-40.0){2}
\node(n3)(16.0,-72.0){$n{-}1$} \node(n4)(56.0,-72.0){3}
\drawedge[ELdist=2.0](n1,n0){} \drawedge[ELdist=1.5](n2,n4){}
\drawedge[ELdist=1.7](n0,n2){}
\drawedge[ELdist=1.7](n3,n1){} \drawedge[ELdist=2.0](n1,n2){}
\put(31,-73){$\dots$}
\end{picture}
\begin{picture}(72,56)(-20,-72)
\gasset{Nw=16,Nh=16,Nmr=8}
\node(n0)(36.0,-16.0){1}
\node(n1)(4.0,-40.0){$n$} \node(n2)(68.0,-40.0){2}
\node(n3)(16.0,-72.0){$n{-}1$} \node(n4)(56.0,-72.0){3}
\drawedge[ELdist=2.0](n1,n0){$b$} \drawedge[ELdist=1.5](n2,n4){$a, b$}
\drawedge[ELdist=1.7](n0,n2){$a, b$}
\drawedge[ELdist=1.7](n3,n1){$a, b$} \drawedge[ELdist=2.0](n1,n2){$a$}
\put(31,-73){$\dots$}
\end{picture}
\end{center}
\caption{The digraph $W_n$ and its unique coloring $\mathrsfs{W}_n$}\label{fig:anan}
\end{figure}

\begin{theorem}
\label{theorem:anan}
The automaton $\mathrsfs{W}_n$ is synchronizing and its reset length is $n^2-3n+3$.
\end{theorem}

\begin{proof}
It is routine to verify that the word $(ab^{n-2})^{n-2}a$, whose length
is $(n-1)(n-2)+1=n^2-3n+3$, is a reset word for $\mathrsfs{W}_n$.

Now let $w$ be a \sw\ for $\mathrsfs{W}_n$ and assume that the length
of $w$ (denoted $|w|$) is minimal. Let $j\in Q=\{1,2,\dots,n\}$ be the
state to which the action of $w$ brings $\mathrsfs{W}_n$. Then from
every state in $Q$ there is a path to $j$ labelled $w$. It is clear
that for each $j\ne 2$ all paths ending at $j$ share the last edge.
Therefore, if $j\ne 2$, removing the last letter from the word $w$
produces a word that still would be a \sw\ for $\mathrsfs{W}_n$.
We conclude that $j=2$ because $|w|$ is minimal.

If $u\in\{a,b\}^*$, the word $uw$ also is a reset word and it also
brings the automaton to the state~2. Hence, for every $\ell\ge|w|$,
there is a path of length $\ell$ in $W_n$ from any given vertex $i$
to~2. In particular, setting $i=2$, we conclude that for every
$\ell\ge|w|$ there is a cycle of length $\ell$ in $W_n$. The digraph
$W_n$ has only two simple cycles: one of length $n$ and one of length $n-1$.
Each cycle of $W_n$ must consist of these two cycles traversed several times
whence each number $\ell\ge|w|$ must be expressible as a non-negative integer
combination of $n$ and $n-1$. Here we invoke Lemma~\ref{sylvester} which
implies that $|w|>n(n-1)-n-(n-1)=n^2-3n+1$. Suppose that $|w|=n^2-3n+2$.
Then there should be a path of this length from the vertex~1 to the vertex~2.
The only outgoing edge of~1 is $(1,2)$, and thus, in the path it must be
followed by a cycle of length $n^2-3n+1$. No cycle of such length may
exist by  Lemma~\ref{sylvester}. Hence $|w|\ge n^2-3n+3$.
\end{proof}

The series $\mathrsfs{W}_n$ was discovered by the first author in 2008
(unpublished). His rather involved proof of Theorem~\ref{theorem:anan}
used a technique developed in~\cite{AVZ}.

As mentioned in Section~\ref{matrices}, Trahtman's recent result~\cite{Tr09}
implies that every primitive digraph admits a synchronizing coloring. This
gives rise to the following natural question: given a primitive digraph on
$n$ vertices, what is the minimum length of \sws\ for its synchronizing
colorings? Observe that in general underlying digraphs of slowly \sa\ may
admit colorings with rather short \sws. Fig.~\ref{fig:cerny} illustrates
this phenomenon: the first coloring of the 4-vertex digraph in Fig.~\ref{fig:cerny}
is the \v{C}ern\'{y} automaton $\mathrsfs{C}_4$ with shortest reset word of
length~9 while the second coloring can be reset of the word $a^3$ of length~3.
Wielandt's digraphs $W_n$, however, can be colored in a essentially unique
way, whence Theorem~\ref{theorem:anan} gives the lower bound $n^2-3n+3$
for the value in question. We strongly believe that this lower bound is
in fact tight, in other words, we suggest a conjecture that is in a sense
parallel to the \v{C}ern\'{y} one.
\begin{conjecture}
\label{hybrid}
Every primitive digraph on $n$ vertices admits a synchronizing coloring
that can be reset by a word of length $n^2-3n+3$.
\end{conjecture}
We observe that while there is a clear analogy between Conjecture~\ref{hybrid}
and the \v{C}ern\'{y} conjecture, the validity of none of them immediately
implies the validity of the other.

Now we discuss a less straightforward way to get a slowly synchronizing series
from Wielandt's digraphs $W_n$. Namely, we aim to show that the \v{C}ern\'{y}
automata $\mathrsfs{C}_n$ are closely related to these digraphs. First, recall
the definition of $\mathrsfs{C}_n$. We may assume that the state set of
$\mathrsfs{C}_n$ is $Q=\{1,2,\dots,n\}$ and the letters $a$ and $b$ act
on $Q$ as follows:
$$\delta(i,a)=\begin{cases}
i &\text{if } i<n,\\
1 &\text{if } i=n;
\end{cases}\quad
\delta(i,b)=\begin{cases}
i+1 &\text{if } i<n,\\
1 &\text{if } i=n.
\end{cases}$$
The automaton $\mathrsfs{C}_n$ is shown in Fig.\,\ref{fig:cerny-n} on the left.

\begin{figure}[ht]
\begin{center}
\unitlength .45mm
\begin{picture}(72,66)(25,-76)
\gasset{Nw=16,Nh=16,Nmr=8}
\node(n0)(36.0,-16.0){1}
\node(n1)(4.0,-40.0){$n$} \node(n2)(68.0,-40.0){2}
\node(n3)(16.0,-72.0){$n{-}1$} \node(n4)(56.0,-72.0){3}
\drawedge[ELdist=2.0](n1,n0){$a,b$} \drawedge[ELdist=1.5](n2,n4){$b$}
\drawedge[ELdist=1.7](n0,n2){$b$}
\drawedge[ELdist=1.7](n3,n1){$b$}
\drawloop[ELdist=1.5,loopangle=30](n2){$a$}
\drawloop[ELdist=2.4,loopangle=-30](n4){$a$}
\drawloop[ELdist=1.5,loopangle=-90](n0){$a$}
\drawloop[ELdist=1.5,loopangle=210](n3){$a$}
\put(31,-73){$\dots$}
\end{picture}
\begin{picture}(72,66)(-25,-76)
\gasset{Nw=16,Nh=16,Nmr=8}
\node(n0)(36.0,-16.0){1}
\node(n1)(4.0,-40.0){$n$} \node(n2)(68.0,-40.0){2}
\node(n3)(16.0,-72.0){$n{-}1$} \node(n4)(56.0,-72.0){3}
\drawedge[ELdist=2.0](n1,n0){$b$} \drawedge[ELdist=1.5](n2,n4){$b,c$}
\drawedge[ELdist=1.7](n0,n2){$b,c$}
\drawedge[ELdist=1.7](n3,n1){$b,c$} \drawedge[ELdist=2.0](n1,n2){$c$}
\put(31,-73){$\dots$}
\end{picture}
\end{center}
\caption{The automaton $\mathrsfs{C}_n$ and the automaton induced by the actions of $b$ and $c=ab$}\label{fig:cerny-n}
\end{figure}
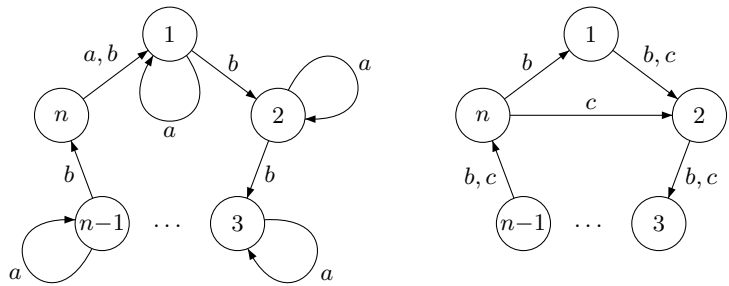

Now we present a new simple proof for the following classic result.
\begin{theorem}[{\mdseries\cite[Lemma~1]{Ce64}}]
\label{theorem:cerny}
The automaton $\mathrsfs{C}_n$ is synchronizing and its reset length is $(n-1)^2$.
\end{theorem}

\begin{proof}
It is easy to see that the word $(ab^{n-1})^{n-2}a$ of length
$n(n-2)+1=(n-1)^2$ is a reset word for $\mathrsfs{C}_n$.

Now let $w$ be a \ssw\ for $\mathrsfs{C}_n$. Since the letter $b$
acts on $Q$ as a cyclic permutation, the word $w$ cannot end with $b$.
(Otherwise removing the last letter gives a shorter \sw.) Thus, we can
write $w$ as $w = w'a$ for some $w'\in\{a,b\}^*$ such that the image
of $Q$ under the action of $w'$ is precisely the set $\{1,n\}$.

Since the letter $a$ fixes each state in its image $\{1,2,\dots,n-1\}$,
every occurrence of $a$ in $w$ except the last one is followed by an
occurrence of $b$. (Otherwise $a^2$ occurs in $w$ as a factor and
reducing this factor to just $a$ results in a shorter \sw.) Therefore,
if we let $c=ab$, then the word $w'$ can be rewritten into a word $v$
over the alphabet $\{b,c\}$. The actions of $b$ and $c$ induce
a new automaton on the state set $Q$; this induced automaton (shown in
Fig.\,\ref{fig:cerny-n} on the right) is obviously isomorphic to the
automaton $\mathrsfs{W}_n$. Since $w'$ and $v$ act on $Q$ in the same way,
the word $vc$ is a \sw\ for the induced automaton. By Theorem~\ref{theorem:anan}
the length of $vc$ (as a word over $\{b,c\}$) is at least $n^2-3n+3$.
Since the action of $b$ on any set $S$ of states cannot change the
cardinality of $S$ and the action of $c$ can decrease the cardinality
by~1 at most, the word $vc$ must contain at least $n-1$ occurrences of
$c$. Hence the length of $v$ over $\{b,c\}$ is at least $n^2-3n+2$ and
$v$ contain at least $n-2$ occurrences of $c$. Since each occurrence of
$c$ in $v$ corresponds to an occurrence of the factor $ab$ in $w'$, we
conclude that the length of $w'$ over $\{a,b\}$ is at least $n^2-3n+2+n-2=n^2-2n$.
Thus, $|w|=|w'a|\ge n^2-2n+1=(n-1)^2$.
\end{proof}

We have found two more series of slowly \sa\ related to Wielandt's digraphs
$W_n$: a series with reset length $n^2-3n+2$ and another one with reset length
$n^2-4n+6$. These two series will be presented in an extended version of
the paper.

Now we discuss a few series related to the digraph $D_n$ defined by the second
matrix in \eqref{wielandt}. The digraph is obtained from $W_n$ by adding
the edge $(n-1,1)$. Fig.\,\ref{fig:dulmage} shows the digraph $D_n$ and
its colorings $\mathrsfs{D}'_n$ and $\mathrsfs{D}''_n$.
\begin{figure}[ht]
\begin{center}
\unitlength .45mm
\begin{picture}(64,56)(30,-72)
\gasset{Nw=16,Nh=16,Nmr=8}
\node(n0)(32.0,-16.0){1}
\node(n1)(0,-40.0){$n$} \node(n2)(64.0,-40.0){2}
\node(n3)(12.0,-72.0){$n{-}1$} \node(n4)(52.0,-72.0){3}
\drawedge[ELdist=2.0](n1,n0){} \drawedge[ELdist=1.5](n2,n4){}
\drawedge[ELdist=1.7](n0,n2){} \drawedge[ELdist=1.7](n3,n0){}
\drawedge[ELdist=1.7](n3,n1){} \drawedge[ELdist=2.0](n1,n2){}
\put(31,-73){$\dots$}
\end{picture}
\begin{picture}(64,56)(0,-72)
\gasset{Nw=16,Nh=16,Nmr=8}
\node(n0)(32.0,-16.0){1}
\node(n1)(0,-40.0){$n$} \node(n2)(64.0,-40.0){2}
\node(n3)(12.0,-72.0){$n{-}1$} \node(n4)(52.0,-72.0){3}
\drawedge[ELdist=2.0](n1,n0){$b$} \drawedge[ELdist=1.5](n2,n4){$a, b$}
\drawedge[ELdist=1.7](n0,n2){$a,b$} \drawedge[ELdist=1.7](n3,n0){$a$}
\drawedge[ELdist=1.7](n3,n1){$b$} \drawedge[ELdist=2.0](n1,n2){$a$}
\put(31,-73){$\dots$}
\end{picture}
\begin{picture}(64,56)(-30,-72)
\gasset{Nw=16,Nh=16,Nmr=8}
\node(n0)(32.0,-16.0){1}
\node(n1)(0,-40.0){$n$} \node(n2)(64.0,-40.0){2}
\node(n3)(12.0,-72.0){$n{-}1$} \node(n4)(52.0,-72.0){3}
\drawedge[ELdist=2.0](n1,n0){$a$} \drawedge[ELdist=1.5](n2,n4){$a, b$}
\drawedge[ELdist=1.7](n0,n2){$a, b$} \drawedge[ELdist=1.7](n3,n0){$a$}
\drawedge[ELdist=1.7](n3,n1){$b$} \drawedge[ELdist=2.0](n1,n2){$b$}
\put(31,-73){$\dots$}
\end{picture}
\end{center}
\caption{The digraph $D_n$ and its colorings $\mathrsfs{D}'_n$ and $\mathrsfs{D}''_n$}\label{fig:dulmage}
\end{figure}
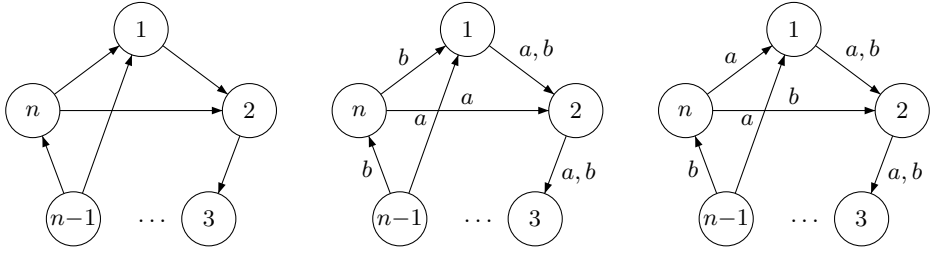

\begin{theorem}
\label{theorem:new series}
The automata $\mathrsfs{D}'_n$ and $\mathrsfs{D}''_n$ are synchronizing with reset lengths
$n^2-3n+4$ and $n^2-3n+2$ respectively.
\end{theorem}

\begin{proof}
The proof of Theorem~\ref{theorem:new series} is similar to that of Theorem~\ref{theorem:anan}.
First we prove that the automaton $\mathrsfs{D}'_n$ is synchronizing and its reset length
is $n^2-3n+4$.

It is routine to verify that the word $(ab^{n - 2})^{n - 2}ba$, whose length
is $(n - 1)(n - 2) + 2 = n^2 -3n + 4$, is a reset word for $\mathrsfs{D}'_n$.

Now let $w$ be a \ssw\ for $\mathrsfs{D}'_n$. Since the letter $b$ acts
as cyclic permutation of $Q$,  the word $w$ ends with $a$. The action of $a$
sends the states $n$ and $1$ to the state 2. Therefore the action of $w$
also brings the automaton to the state $2$ and so does the action of
the word $uw$ for any $u\in\{a,b\}^*$. Hence, for every $\ell\ge|w|$,
there is a path of length $\ell$ in $D_n$ from any given vertex $i$
to~2. In particular, setting $i=2$, we conclude that for every $\ell\ge|w|$
there is a cycle of length $\ell$ in $D_n$. The digraph $D_n$ has only
three simple cycles: one of length $n$ and two of length $n-1$.
Since each cycle of $D_n$ must be composed from these three cycles,
every number $\ell\ge|w|$ must be expressible as a non-negative integer combination
of $n$ and $n-1$. Then Lemma~\ref{sylvester} implies that $|w|>n(n-1)-n-(n-1)=n^2-3n+1$.

Suppose that $|w|=n^2-3n+2$. Then in $D_n$ there should be a path of this length
from the vertex~1 to the vertex~2. The only outgoing edge of~1 is $(1,2)$, and thus,
in the path it must be followed by a cycle of length $n^2-3n+1$. No cycle of such
length may exist by Lemma~\ref{sylvester}. Hence $|w|\ge n^2-3n+3$.

Finally, suppose that $|w|=n^2-3n+3$. There should be a path of this length
from the vertex~$n-1$ to the vertex~2. Since $b$ is a cyclic permutation,
the first letter of $w$ must be $a$ but $a$ sends $n-1$ to $1$. Therefore
in $D_n$ there is a path of length $n^2-3n+2$ from the vertex~$1$ to the vertex~2.
In the previous paragraph we have shown that this is not possible. Hence $|w|\ge n^2-3n+4$.

\medskip

Now we show that the automaton $\mathrsfs{D}''_n$ is synchronizing and
its reset length is $n^2-3n+2$.

Again, the word $(ba^{n - 1})^{n - 3}ba$ of length $n(n - 3) + 2 = n^2 -3n + 2$
is easily seen to be a reset word for $\mathrsfs{D}''_n$.

Now let $w$ be a \sw\ for $\mathrsfs{D}''_n$. Then, as above, each
$\ell\ge|w|$ must be the length of a cycle in the digraph $D_n$
whence $|w|>n(n-1)-n-(n-1)=n^2-3n+1$ by Lemma~\ref{sylvester}.
\end{proof}

The series $\mathrsfs{D}'_n$ is of interest because
for $n>6$ it yields the maximum known value of reset length beyond the \v{C}ern\'y series
$\mathrsfs{C}_n$ and also the maximum known value of reset length for \sa\ without loops.
The series $\mathrsfs{D}''_n$ also enjoys an extremal property: it provides the maximum
known value of reset length for \sa\ in which no letter acts as a permutation.

One more series of slowly \sa\ related to the digraphs $D_n$ has reset length
$n^2-4n+6$. It will be presented in an extended version of this paper.

Except the \v{C}ern\'y series $\mathrsfs{C}_n$, the only infinite series of
2-letter slowly \sa\ published so far was the series $\mathrsfs{B}_n$ ($n>3$
is odd) constructed in~\cite{AVZ}. The automaton $\mathrsfs{B}_n$ has $Q=\{1,2,\dots,n\}$
as its state set, and its input letters $a$ and $b$ act on $Q$ as follows:
$$\delta(i,a)=\begin{cases}
i &\text{if } i<n-1,\\
1 &\text{if } i=n-1,\\
2 &\text{if } i=n;
\end{cases}\quad
\delta(i,b)=\begin{cases}
i+1 &\text{if } i<n,\\
1 &\text{if } i=n.
\end{cases}$$
The automaton $\mathrsfs{B}_n$ is shown in Fig.\,\ref{fig:avz-n} on the left.

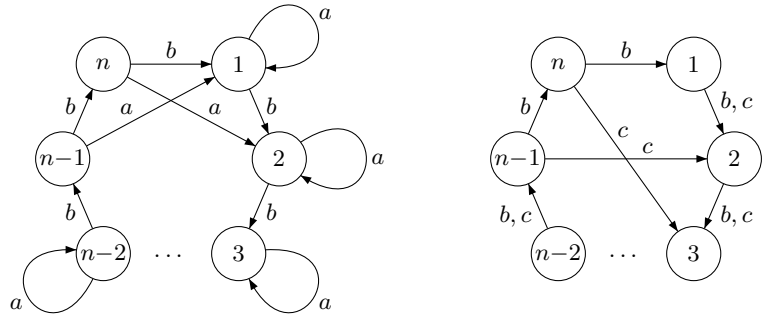
\begin{figure}[ht]
\begin{center}
\unitlength .45mm
\begin{picture}(72,66)(30,-76)
\gasset{Nw=16,Nh=16,Nmr=8}
\node(n1)(56.0,-16.0){1}
\node(n2)(68.0,-44.0){2}
\node(n3)(56.0,-72.0){3}
\node(n4)(16.0,-72.0){$n{-}2$}
\node(n5)(4.0,-44.0){$n{-}1$}
\node(n6)(16.0,-16.0){$n$}
\drawedge[ELdist=1.7](n1,n2){$b$}
\drawedge[ELdist=1.7](n2,n3){$b$}
\drawedge[ELdist=1.7](n4,n5){$b$}
\drawedge[ELdist=1.7](n5,n6){$b$}
\drawedge[ELdist=1.7](n6,n1){$b$}
\drawedge[ELdist=1.7,ELpos=40](n5,n1){$a$}
\drawedge[ELdist=1.7,ELpos=60](n6,n2){$a$}
\drawloop[ELdist=1.5,loopangle=30](n1){$a$}
\drawloop[ELdist=1.5,loopangle=0](n2){$a$}
\drawloop[ELdist=1.5,loopangle=-30](n3){$a$}
\drawloop[ELdist=1.5,loopangle=210](n4){$a$}
\put(31,-73){$\dots$}
\end{picture}
\begin{picture}(72,66)(-30,-76)
\gasset{Nw=16,Nh=16,Nmr=8}
\node(n1)(56.0,-16.0){1}
\node(n2)(68.0,-44.0){2}
\node(n3)(56.0,-72.0){3}
\node(n4)(16.0,-72.0){$n{-}2$}
\node(n5)(4.0,-44.0){$n{-}1$}
\node(n6)(16.0,-16.0){$n$}
\drawedge[ELdist=1.7](n1,n2){$b,c$}
\drawedge[ELdist=1.7](n2,n3){$b,c$}
\drawedge[ELdist=1.7](n4,n5){$b,c$}
\drawedge[ELdist=1.7](n5,n6){$b$}
\drawedge[ELdist=1.7](n6,n1){$b$}
\drawedge[ELdist=1.7,ELpos=60](n5,n2){$c$}
\drawedge[ELdist=1.7,ELpos=40](n6,n3){$c$}
\put(31,-73){$\dots$}
\end{picture}
\end{center}
\caption{The automaton $\mathrsfs{B}_n$ and the automaton induced by the actions of $b$ and $c=ab$}\label{fig:avz-n}
\end{figure}

\begin{theorem}[{\mdseries\cite[Theorem~1.1]{AVZ}}]
\label{theorem:avz}
If $n>3$ is odd, then the automaton $\mathrsfs{B}_n$ is synchronizing and
its reset length is $n^2-3n+2$.
\end{theorem}

The proof of Theorem~\ref{theorem:avz} in \cite{AVZ} is quite involved. Now we can easily
prove this result using an argument similar to that in our proof of Theorem~\ref{theorem:cerny}.
The key observation is that the automaton induced by the actions of $b$ and $c=ab$ on the set $Q$
as shown in Fig.\,\ref{fig:avz-n} on the right is nothing but a coloring of one of the digraphs
with exponent $n^2-3n+2$, namely, of the digraph defined by the second matrix in \eqref{odd island}.
The details of the proof will appear in an extended version of this paper.

\section{Experiments}
\label{experiments}

Here we briefly describe the settings of our experiments. Recall that a DFA
$\mathrsfs{A}=\langle Q,\Sigma,\delta\rangle$ is said to be \emph{initially-connected}
if there exists a state $q_0\in Q$ from which every state $q\in Q$ is reachable,
that is, $q=\delta(q_0,w)$ for some $w\in\Sigma^*$. In general a \san\ need not
be initially-connected. However, it is well known that when studying the \v{C}ern\'{y}
conjecture, we may restrict ourselves to DFA whose underlying digraphs are strongly
connected because the validity of the conjecture can be easily reduced to this case
(see \cite{Pi78} for example). Clearly, DFA with strongly connected underlying digraphs
are initially-connected.

We used a convenient string representation of initially-connected DFA (ICDFA)
developed in~\cite{AMR} to generate all such DFA with up to 9 states and 2 input
letters. Each ICDFA was tested for synchronizability and then for each \san\
its reset length was calculated. For these tasks, we implemented standard
algorithms (see~\cite{Sa05,Vo08}) in C.

The main difficulty that had to be overcome is that the number of ICDFA dramatically
grows with the number of states. (For 9 states, there are about 700 billions
ICDFA with 2 input letters.) The problem, however, can be efficiently parallelized.
For this, a dedicated processor was programmed to generate ICDFA in portions (slices
in terminology of~\cite{AMR}) that were fed to other processors for synchronization
tests etc. The management program was written in C with MPI. Calculations organized
this way took less than a day of running a small size computer grid based on a number
of AMD Opteron 2.6 GHz processors.

All slowly synchronizing automata that we found were double-checked by running on
them the package TESTAS developed by Trahtman~\cite{Tr06}.

Our experiments have also produced some interesting statistical results that will
be discussed elsewhere.


\begin{thebibliography}{99}
\bibitem{AGW}
Adler, R.L., Goodwyn, L.W., Weiss, B.: Equivalence of topological Markov shifts.
Israel J. Math. 27, 49--63 (1977)

\bibitem{AS09}
Almeida, J.; Steinberg, B.: Matrix mortality and the \v{C}ern\'{y}--Pin
conjecture. In:  Diekert, V.; Nowotka, D. (eds.), Developments in
Language Theory, Lect.\ Notes Comput.\ Sci., vol.\,5583, pp. 67--80.
Springer, Heidelberg (2009)

\bibitem{AMR}
Almeida, M.; Moreira, N.; Reis, R.: Enumeration and generation with a string
automata representation. Theor.\ Comput.\ Sci. 387, 93--102 (2007)

\bibitem{AVZ}
Ananichev, D.S., Volkov, M.V., Zaks, Yu.I.: Synchronizing automata
with a letter of deficiency 2. Theor.\ Comput.\ Sci. 376, 30--41 (2007)

\bibitem{Ce64}
\v{C}ern\'{y}, J.: Pozn\'{a}mka k homog\'{e}nnym eksperimentom s
kone\v{c}n\'{y}mi automatami. Matematicko-fyzikalny \v{C}asopis
Slovensk.\ Akad.\ Vied 14(3) 208--216 (1964) (in Slovak)

\bibitem{Du98}
Dubuc, L.: Sur les automates circulaires et la conjecture de
\v{C}ern\'y. RAIRO Inform.\ Th\'eor.\ Appl. 32, 21--34 (1998) (in
French)

\bibitem{DM62}
Dulmage, A.L., Mendelsohn, N.S.: The exponent of a primitive matrix.
Can.\ Math.\ Bull. 5, 241--244 (1962)

\bibitem{DM64}
Dulmage, A.L., Mendelsohn, N.S.: Gaps in the exponent set of primitive
matrices. Ill.\ J. Math. 8, 642--656 (1964)

\bibitem{Ep90}
Eppstein, D.: Reset sequences for monotonic automata. SIAM J.
Comput. 19, 500--510 (1990)

\bibitem{Hi88}
Higgins, P.M.: The range order of a product of $i$ transformations
from a finite full transformation semigroup, Semigroup Forum 37, 31--36
(1988)

\bibitem{Ka03}
Kari, J.: Synchronizing finite automata on Eulerian digraphs.
Theoret.\ Comput.\ Sci. 295, 223--232 (2003)

\bibitem{Pi78}
Pin, J.-E.: Le probl\`eme de la synchronization et la conjecture de
\v{C}ern\'y. Th\`ese de 3\`eme cycle. Universit\'e de Paris 6 (1978) (in
French)

\bibitem{Pi83}
Pin, J.-E.: On two combinatorial problems arising from automata
theory. Ann.\ Discrete Math. 17, 535--548 (1983)

\bibitem{RaAl05}
Ram\'{\i}rez Alfons\'{\i}n, J.L.: The diophantine Frobenius problem.
Oxford University Press (2005)

\bibitem{Sa05}
Sandberg, S.: Homing and synchronizing sequences. In: Broy, M.
et~al. (eds.), Model-Based Testing of Reactive Systems. Lect.\
Notes Comput.\ Sci., vol.\,3472, pp.\,5--33. Springer, Heidelberg
(2005)

\bibitem{SZ}
Skvortsov E., Zaks Yu.: Synchronizing random automata. Submitted;
proceedings version in: Rigo, M. (ed.), AutoMathA 2009, Universit\'e
de Li\`ege (2009)

\bibitem{Tr06}
Trahtman, A.N.: An efficient algorithm finds noticeable trends
and examples concerning the \v{C}ern\'y conjecture. In:
Kr\'alovi\v{c}, R.; Urzyczyn, P. (eds.), Mathematical Foundations
of Computer Science. Lect.\ Notes Comput.\ Sci., vol.\,4162, pp.\,789--800
Springer, Heidelberg (2006)

\bibitem{Tr06a}
Trahtman, A.N.: Notable trends concerning the synchronization
of graphs and automata. Electr.\ Notes Discrete Math. 25, 173--175 (2006)

\bibitem{Tr07}
Trahtman, A.N.: The \v{C}ern\'y conjecture for aperiodic automata.
Discrete Math.\ Theor.\ Comput.\ Sci. 9(2), 3--10 (2007)

\bibitem{Tr08}
Trahtman, A.N.: Some aspects of synchronization of DFA. J. Comput.\ Sci.\
Technol. 23, 719--727 (2008)

\bibitem{Tr09}
Trahtman, A.N.: The Road Coloring Problem. Israel J. Math.
\textbf{172}, 51--60 (2009)

\bibitem{Vo08}
Volkov, M.V.: Synchronizing automata and the \v{C}ern\'{y}
conjecture. In: Mart\'\i{}n-Vide, C.; Otto, F.; Fernau, H. (eds.),
Languages and Automata: Theory and Applications. Lect.\ Notes
Comput.\ Sci., vol.\,5196, pp.\,11--27.  Springer, Heidelberg (2008)

\bibitem{Vo09}
Volkov, M.V.: Synchronizing automata preserving a chain of partial
orders. Theoret.\ Comput.\ Sci. 410, 2992--2998 (2009)

\bibitem{Wi50}
Wielandt, H.: Unzerlegbare, nicht negative Matrizen. Math.\ Z.
52, 642--648 (1950) (in German)
\end{thebibliography}
\end{document}